\documentclass[11pt]{article}
\usepackage{amssymb,amsmath}
\usepackage{float}
\usepackage{color,fullpage}
\usepackage{amsthm}
\newtheorem{theorem}{Theorem}
\newtheorem{corollary}{Corollary}
\newtheorem{lemma}{Lemma}

\newtheorem{remark}{Remark}

\begin{document}

\title{ New bounds for circulant Johnson-Lindenstrauss embeddings}
\author{
Hui Zhang\thanks{Department of Mathematics and Systems Science,
College of Science, National University of Defense Technology,
Changsha, Hunan, China, 410073. Email: \texttt{hhuuii.zhang@gmail.com}}
\and Lizhi Cheng\thanks{Department of Mathematics and Systems Science,
College of Science, National University of Defense Technology,
Changsha, Hunan, China, 410073. Email: \texttt{clzcheng@nudt.edu.cn}}}
\date{\today}

\maketitle





\begin{abstract}
This paper analyzes circulant Johnson-Lindenstrauss (JL) embeddings which, as an important class of structured random JL embeddings, are formed by randomizing the column signs of a circulant matrix generated by a random vector. With the help of recent decoupling techniques and matrix-valued Bernstein inequalities, we obtain a new bound $k=O(\epsilon^{-2}\log^{(1+\delta)} (n))$ for Gaussian circulant JL embeddings. Moreover, by using the Laplace transform technique (also called Bernstein's trick), we extend the result to subgaussian case. The bounds in this paper offer a small improvement over the current best bounds for Gaussian circulant JL embeddings for certain parameter regimes and are derived using more direct methods.
\end{abstract}


\section{Introduction}\label{intro}
The Johnson-Lindenstrauss (JL) lemma \cite{JL} is by now a standard technique in high dimensional data processing. The lemma shows the existence, with high probability, of JL embeddings, or linear maps $A \in \mathbb{R}^d \rightarrow \mathbb{R}^k$ (with $k < d$) which embed a fixed set of $n$ points $\{x_1\cdots x_n\}\subset \mathbb{R}^d$ into $\mathbb{R}^k$ with distortion at most $\epsilon$. The best known embedding dimension $k$, as is achieved by e.g., Gaussian random matrices, is $k=O(\epsilon^{-2}\log(n))$. Recently, there is growing interest in analyzing structured random JL embeddings which, unlike Gaussian random matrices, have fast matrix-vector multiplication routines. In this paper, we focus on circulant JL embeddings which, as an important class of such structured random JL embeddings, are formed by randomizing the column signs of a circulant matrix generated by a random vector. The first result for circulant JL embeddings might be formulated as follows:

\begin{theorem}\label{theo1}(\cite{Hinrichs})
Let $x^1,x^2,\cdots, x^n$ be $n$ points in the $d$-dimensional
Euclidean space $\mathbb{R}^d$. Let $\epsilon \in (0,\frac{1}{2})$
and let $k=O(\epsilon^{-2}\log^3(n))$ be a natural number. Assume that
$f$ is a composition of a $k\times d$ random circulant matrix
$M_{a,k}$ with a $d\times d$ random diagonal matrix
$D_{\varkappa}$, i.e., $f(x)=\frac{1}{\sqrt{k}}M_{a,k}
D_{\varkappa}x$. Then with probability at least 2/3 the following
holds
\begin{equation}
(1-\epsilon)\|x^i-x^j\|_2^2\leq \|f(x^i)-f(x^j)\|_2^2 \leq
(1+\epsilon)\|x^i-x^j\|_2^2,~~ i,j,=1,\cdots, n.
\end{equation}
\end{theorem}
Here, the random circulant matrix $M_{a,k}$ is defined by random vector $a=(a_0,\cdots, a_{d-1})$ whose entries are independent Bernoulli variables
or independent normally distributed variables. Concretely,
$$
M_{a,k} =\left(\begin{array}{ccccc}
a_0 & a_1 & a_2 & \cdots & a_{d-1}\\
a_{d-1} & a_0 & a_1 & \cdots & a_{d-2}\\
\vdots & \vdots & \vdots & \ddots & \vdots\\
a_{d-k+1} & a_{d-k+2} & a_{d-k+3} & \cdots &
a_{d-k}\end{array}\right)\in \mathbb{R}^{k \times d}.
$$
The random diagonal matrix $D_{\varkappa}$ is
$$
D_{\varkappa}=\left(\begin{array}{ccccc}
\varkappa_0 & 0 & 0 & \cdots & 0\\
0 & \varkappa_1 & 0 & \cdots & 0\\
\vdots & \vdots & \vdots & \ddots & \vdots\\
0 & 0 & 0 & \cdots & \varkappa_{d-1}\end{array}\right)\in \mathbb{R}^{d \times d},
$$
where
$\varkappa=(\varkappa_0,\varkappa_1,\cdots,\varkappa_{d-1})$ is
a Bernoulli sequence, i.e., each entry of $\varkappa$ takes the
values $+1$ or $-1$ with probability $1/2$.  Here and thereafter, we will call the mapping $f$ (sub)gaussian circulant JL embedding when random vector $a$ is set as (sub)gaussian random vector.

Compared with the standard bound $k=O(\epsilon^{-2}\log (n))$, Theorem \ref{theo1} only established a worse bound $k=O(\epsilon^{-2}\log^3 (n))$.
Later on, Vyb\'{\i}ral \cite{Vyb} improved the bound to
$k=O(\epsilon^{-2}\log^2 (n))$ by employing the discrete Fourier
transform and singular value decomposition to deal with the
dependence caused by the circulant structure. Recently, by randomizing the column signs of matrices that have Restricted Isometry Property (RIP)\cite{Cand},  Krahmer and Ward \cite{Krahmer} further
improved the bound to $k=O(\epsilon^{-2}\log(n)\log^4(d))$ which is better than another recent bound
$k=O(\epsilon^{-4}\log(n)\log^4(d))$ by Ailon and Liberty \cite{Ailon}. Most recently, Krahmer, Mendelson, and Rauhut \cite{KMR} derived new bounds for the RIP of partial circulant matrices. By combining these bounds with the connection between RIP and JL in \cite{Krahmer}, the current best bound for Gaussian circulant JL embedding reads
\begin{equation}\label{best}
k=O(\epsilon^{-2}\log(n)(\log d)^2(\log\log d)^2).
\end{equation}
 We summarize these JL bounds in Table 1.

\begin{table}[ht]
\begin{center}\begin{tabular}{c|c}\hline
work & JL bound  \\\hline
\cite{Hinrichs}    & $k=O(\epsilon^{-2}\log^3 (n))$ \\\hline
\cite{Vyb}      & $k=O(\epsilon^{-2}\log^2 (n))$ \\\hline
\cite{Ailon}    & $k=O(\epsilon^{-4}\log(n)\log^4(d))$\\\hline
\cite{Krahmer}  & $k=O(\epsilon^{-2}\log(n)\log^4(d))$ \\\hline
\cite{Krahmer}, \cite{KMR}  & $k=O(\epsilon^{-2}\log(n)(\log d)^2(\log\log d)^2)$ \\\hline
\end{tabular}\end{center}
\caption{Bounds for Gaussian circulant JL embeddings}\label{JL}
\end{table}


\subsection{Main results}
In this study, we combine the decoupling technique in \cite{Vyb} with the matrix-value Bernstein inequality in \cite{Tro1} to derive a new and improved bound  $k=O(\epsilon^{-2}\log^{(1+\delta)} (n))$ for Gaussian circulant JL embeddings.

 Traditionally, the key step in JL lemma is to estimate the probablity bounds of $\mathbb{P}(\|f(x)\|_2^2\geq (1+\epsilon)k)$ and $\mathbb{P}(\|f(x)\|_2^2\leq (1-\epsilon)k)$. The authors in \cite{Hinrichs} obtained the following estimations:
\begin{equation}\label{ineq01}
\mathbb{P}(\|f(x)\|_2^2\geq (1+\epsilon)k)\leq
\exp(-c(k\epsilon^2)^{1/3})
\end{equation}
and
\begin{equation}\label{ineq02}
\mathbb{P}(\|f(x)\|_2^2\leq (1-\epsilon)k)\leq \exp(-c(k\epsilon^2)^{1/3}),
\end{equation}
where $c$ is an absolute constant. One can see that it is just the power $1/3$ making the bound to be $k\epsilon^2\sim\log^3 (n)$,
 i.e., $k=O(\epsilon^{-2}\log^3 (n))$.
Vyb\'{\i}ral \cite{Vyb} improved the right-hand side of inequalities (\ref{ineq01}) and (\ref{ineq02}) to $\exp(-\frac{ck\epsilon^2}{\log n})$, and hence directly derived a better bound $k=O(\epsilon^{-2}\log^2 (n))$. Our main result, stated in Theorem \ref{lem0}, is more general and can recover the result in \cite{Vyb} under a strictly weaker constraint on number $n$ if $d>12$. Also, the bound for Gaussian circulant JL embeddings, derived in Corollary \ref{cor1}, offers an improvement over existing bounds.
\begin{theorem}[Main result]\label{lem0} Let $k\leq d$ be natural numbers and let $\epsilon \in
(0,\frac{1}{2})$. Let $x \in \mathbb{R}^d$ be a unit vector,
$a=(a_0,a_1,\cdots,a_{d-1})\sim N_d(0, I_d)$. Assume that
$f$ is a composition of a $k\times d$ Gaussian circulant matrix
$M_{a,k}$ with a $d\times d$ random diagonal matrix
$D_{\varkappa}$, i.e., $f(x)= M_{a,k}
D_{\varkappa}x$. Then with probability at least $1-(d+k)e^{-\frac{\tau\log^\delta n}{2}}$ it holds
\begin{equation}\label{inequ1}
\mathbb{P}(\|f(x)\|_2^2\geq (1+\epsilon)k)\leq
\exp\left(-\frac{c(\tau)k\epsilon^2}{\log^{\delta} n}\right)
\end{equation}
and
\begin{equation}\label{inequ2}
\mathbb{P}(\|f(x)\|_2^2\leq (1-\epsilon)k)\leq \exp\left(-\frac{c(\tau)k\epsilon^2}{\log^{\delta} n}\right),
\end{equation}
where $c(\tau)=\frac{1}{8\tau}$, $\delta$ and $\tau$ are positive parameters.
 \end{theorem}

\begin{remark}
Setting $\tau=2$ and $\delta=1$ in Theorem \ref{lem0} gives $c(\tau)=\frac{1}{16}$ and
$$1-(d+k)e^{-\frac{\tau\log^\delta n}{2}}=1-\frac{d+k}{n}.$$
Thus, letting $1-\frac{d+k}{n}\geq \frac{5}{6}$, i.e., $n\geq \sqrt{6(d+k)}$,
Theorem \ref{lem0} rederives the inequalities (3.4) and (3.5) in \cite{Vyb} with an explicit value $c=\frac{1}{16}$, and the condition on $n$ is strictly relaxed from $n\geq d$ to $n\geq \sqrt{6(d+k)}$ when $d>12$ since it holds $n\geq d>\sqrt{6(d+k)}$ for any $k<d$; for more details please refer to Lemma 3.1 and the proof of Theorem 1.3 in \cite{Vyb}.
\end{remark}

The following corollary follows by a union bound over $C^2_n$ pairs of points; note that we have set $\tau=2$ for simplicity.

\begin{corollary}\label{cor1}
Let $x^1,x^2,\cdots, x^n$ be $n$ points in the $d$-dimensional
Euclidean space $\mathbb{R}^d$. Let $\epsilon \in (0,\frac{1}{2})$
and let $k=O(\epsilon^{-2}\log^{(1+\delta)}(n))$ be a natural number, where $\delta>0$. Assume that
$f$ is a composition of a $k\times d$ Gaussian circulant matrix
$M_{a,k}$ with a $d\times d$ random diagonal matrix
$D_{\varkappa}$, i.e., $f(x)=\frac{1}{\sqrt{k}}M_{a,k}
D_{\varkappa}x$. Then with probability at least $\frac{2}{3}\left(1-(d+k)e^{-\log^\delta n}\right)$,  the following
holds
\begin{equation}
(1-\epsilon)\|x^i-x^j\|_2^2\leq \|f(x^i)-f(x^j)\|_2^2 \leq
(1+\epsilon)\|x^i-x^j\|_2^2,~~i,j,=1,\cdots, n.
\end{equation}
\end{corollary}

\begin{remark}
Compared with the current best bound (\ref{best}), Corollary \ref{cor1} only offers an improved bound for a relatively small non asymptotic range of $n$. In fact, in order for the stated probability to be positive, we derive that $\log(d+k)<\log^\delta(n)$ and hence $\log(d)<\log^\delta (n)$. On the other hand, we need $\log^\delta (n)\leq \log^2(d)$ to have an improved estimate over (\ref{best}). Therefore, for the parameter regimes satisfying $$\log(d)< \log^\delta (n)\leq \log^2(d),$$  Corollary \ref{cor1} indeed offers an improved bound $k=O(\epsilon^{-2}\log^{(1+\delta)}(n))$ for Gaussian circulant JL embedding. However, once $n$ is sufficiently large that $\log^2(d)\leq \log^\delta(n)$, the derived bound in Corollary \ref{cor1} becomes increasingly worse than (\ref{best}). In other words, bound (\ref{best}) is asymptotically stronger than that in Corollary \ref{cor1}.
\end{remark}

\subsection{Extension}
We generalize the main result to the case of subgaussian circulant JL embedding by borrowing the Laplace transform technique (also called Bernstein's trick). Here $X$ is a subgaussian random variable with constant $\eta$ referring to $E[\exp(tX)]\leq \exp(\eta t^2)$ for some $\eta>0$. We only discuss the case of $\eta\leq 1/2$ where includes many types of random circulant matrices we are interested in. An important type is the Bernoulli circulant matrix. In fact, if $X$ is a Bernoulli random variable, then $E [\exp (tX)] = \frac{1}{2}\exp(t)+\frac{1}{2}\exp(-t) = cosh(t) \leq \exp(\frac{1}{2}t^2)$. So the Bernoulli random variable $X$ is subgaussian with $\eta = \frac{1}{2}$.  For the subgaussian case, we have the following results:

\begin{theorem}\label{lem01}
 Let $k\leq d$ be natural numbers and let $\epsilon \in
(0,\frac{1}{2})$. Let $x \in \mathbb{R}^d$ be a unit vector,
subgaussian vector $a=(a_0,a_1,\cdots,a_{d-1})$
having a uniform subgaussian constant $\eta>0$. Assume that
$f$ is a composition of a $k\times d$ subgaussian circulant matrix
$M_{a,k}$ with a $d\times d$ random diagonal matrix
$D_{\varkappa}$, i.e., $f(x)= M_{a,k}
D_{\varkappa}x$. Then with probability at least $1-(d+k)e^{-\frac{\tau\log^\delta n}{2}}$ it holds
\begin{equation}\label{inequ01}
\mathbb{P}(\|f(x)\|_2^2\geq (1+\epsilon)k)\leq
\exp\left(-\frac{c(\theta, \eta, \tau)k\epsilon^2}{\log^{2\delta} n}\right)
\end{equation}
and
\begin{equation}\label{inequ02}
\mathbb{P}(\|f(x)\|_2^2\leq (1-\epsilon)k)\leq \exp\left(-\frac{c(\theta, \eta, \tau)k\epsilon^2}{\log^{2\delta} n}\right).
\end{equation}
Here, $c(\theta, \eta, \tau)=\theta(\frac{1}{2\tau\eta}-4\theta)$ is some absolute constant, where $0<\theta <\min\{1,\frac{1}{8\eta\tau}\}$ and $\delta, \tau >0$ are fixed parameters,  the number $n$ needs to be set big enough such that $\frac{2\theta\epsilon}{\log^\delta n}<\frac{1}{2}$, and the subgaussian constant $\eta$ obeys $\frac{1}{2}\frac{1-\beta^2}{1+\beta^2}\leq \eta \leq \frac{1}{2}$ with $\beta=\frac{\theta\epsilon}{\tau\log^{2\delta}n}<\frac{1}{2}$.
\end{theorem}

Again, the following corollary follows by a union bound over $C^2_n$  pairs of points and setting $\tau=2$.

\begin{corollary}\label{cor2}
Let $x^1,x^2,\cdots, x^n$ be $n$ points in the $d$-dimensional
Euclidean space $\mathbb{R}^d$. Let $\epsilon \in (0,\frac{1}{2})$
and let $k=O(\epsilon^{-2}\log^{1+2\delta}n)$ be a natural number, where $\delta$ is a fixed positive parameter. Assume that
$f$ is a composition of a $k\times d$ subgaussian circulant matrix
$M_{a,k}$ with a $d\times d$ random diagonal matrix
$D_{\varkappa}$, i.e., $f(x)=\frac{1}{\sqrt{k}}M_{a,k}
D_{\varkappa}x$. Assume that the subgaussian constant $\eta$ obeys $\frac{1}{2}\frac{1-\beta^2}{1+\beta^2}\leq \eta \leq \frac{1}{2}$ where $\beta=\frac{\theta\epsilon}{\log^{2\delta}n}<1$, $0<\theta <\min\{1,\frac{1}{16\eta}\}$, and $n$ is big enough such that $\frac{2\theta\epsilon}{\log^\delta n}<\frac{1}{2}$. Then with probability at least $\frac{2}{3}\left(1-(d+k)e^{-\log^\delta (n)}\right)$ the following
holds
\begin{equation}
(1-\epsilon)\|x^i-x^j\|_2^2\leq \|f(x^i)-f(x^j)\|_2^2 \leq
(1+\epsilon)\|x^i-x^j\|_2^2,~~i,j,=1,\cdots, n.
\end{equation}
\end{corollary}

\begin{remark}
The bound $k=O(\epsilon^{-2}\log^{1+2\delta}n)$ is independent of the parameters $\eta, \beta$ and $\theta$. These parameters are used to bound the subgaussian constant. In other words, the conclusion in Corollary \ref{cor2} only applies to some special subgaussian cases.
\end{remark}

\begin{remark}
Although our main result can be extended to subgaussian case, we have to admit that the bound $k=O(\epsilon^{-2}\log^{1+2\delta}(n))$ in Corollary \ref{cor2} is weaker than (\ref{best}) due to the factor $\log^{2\delta}(n)$ and the implicit requirement $\log(d)<\log^\delta (n)$ in the probability bound. However, our analysis is more direct than that in \cite{KMR} and our bound is comparable to (\ref{best}) when the number of points $n$ is approximately the same as the ambient dimension $d$.
\end{remark}

\section{Proof of Theorem \ref{lem0}}
In this section, we will prove Theorem \ref{lem0} by showing that for any fixed unit vector $x$, $f(x)= M_{a,k}
D_{\varkappa}x$ has the concentration property. We divide the proof of Theorem \ref{lem0} into three steps. Since the random matrix $M_{a,k}
D_{\varkappa}$ couples the random vectors $a$ and $\varkappa$ together, the first step decouples these two random vectors so that we can apply some existing concentration results to them separately. The second step estimates the spectral norm of random matrix $Y$ whose randomness is from the Bernoulli random vector $\varkappa$. By using the special structure of the random matrix $Y$, we deduce a tighter and more general estimate than that from \cite{Vyb}. Our derivation relies on the matrix-valued Berstein inequality in \cite{Tro1}. The last step is a direct application of the concentration of quadratic function to the Gaussian random vector $a$.

\textbf{Step 1: Decoupling.}
We define matrix
$$
Y=\left(\begin{array}{ccccc}
x_0\varkappa_0 & x_1\varkappa_1 & x_2\varkappa_2 & \cdots & x_{d-1}\varkappa_{d-1}\\
x_1\varkappa_1 & x_2\varkappa_2 & x_3\varkappa_3 & \cdots & x_0\varkappa_0\\
\vdots & \vdots & \vdots & \ddots & \vdots\\
x_{k-1}\varkappa_{k-1} & x_{k}\varkappa_k& x_{k+1}\varkappa_{k+1} & \cdots &
x_{k-2}\varkappa_{k-2}\end{array}\right)\in \mathbb{R}^{k \times d}.
$$
Then it holds
$$\|f(x)\|_2^2=\|M_{a,k}D_{\varkappa}x\|^2_2=\|Ya\|^2_2.$$
Let $Y=U\Sigma V^T$ be the singular value decomposition of $Y$. Since $Y\in \mathbb{R}^{k\times d}$, we take matrices $U\in \mathbb{R}^{k\times k} , V\in \mathbb{R}^{d\times k}$ to be real orthogonal matrices \cite{Horn}.  Thus $b=V^Ta$ is a $k-$dimensional vector of independent
Gaussian variable. Hence,
$$\|Ya\|^2_2=\|U\Sigma V^Ta\|^2_2=\|U\Sigma b\|^2_2=\|\Sigma b\|^2_2=\sum_{j=0}^{k-1}|\lambda_j|^2b_j^2,$$ where $\lambda_j, j=0,1,\cdots, k-1$
are the singular values of $Y$, and $b_j= \sum_{i=0}^{d-1}V_{ij}a_i$. Let $\mu_j=|\lambda_j|^2$. Then
\begin{equation}\label{equ1}
\|\mu\|_1=\sum_{j=0}^{k-1}|u_j|=\sum_{j=0}^{k-1}|\lambda_j|^2=\|Y\|_F^2=k,
\end{equation}
where $\|Y\|_F$ is the Frobenius norm of $Y$, and the last identity is due to that $x \in \mathbb{R}^d$ is a unit vector.

\textbf{Step 2: Spectral estimate.}
While the analysis of decoupling process in the first step closely follows from Vyb\'{\i}ral \cite{Vyb}, the estimate of the spectral norm of $Y$ is quite different.  We begin with the following
lemma \cite{Tro1}:


\begin{lemma}[Matrix-valued Bernstein inequality]\label{lem1}
Consider a finite sequence $\{B_i\}$ of fixed matrices with
dimension $d_1\times d_2$, and let $\{\xi_i\}$ be a finite sequence
of independent standard normal variables or symmetrical Bernoulli
variables. Then, for all $t\geq 0$,
\begin{equation}\label{equ2}
\mathbb{P} \{\|\sum_i\xi_iB_i\|\geq t \}\leq
(d_1+d_2)e^{-t^2/2\sigma^2},
\end{equation}
where $$\sigma^2:=\max\{\|\sum_iB_iB^T_i\|,\|\sum_iB^T_iB_i\|\}.$$

\end{lemma}

To apply this lemma to our case, we define two $d\times d$
permutation matrices

~~~~~~~~~~~~$$ P=\left(\begin{array}{cccccc}
0 & 0 & 0 & \cdots & 0  & 1\\
1 & 0 & 0 & \cdots & 0  & 0\\
0 & 1 & 0 & \cdots & 0  & 0\\
\vdots & \vdots & \vdots &  \ddots & \vdots  & \vdots\\
0 & 0 & 0 & \cdots & 1  & 0 \end{array}\right)\in \mathbb{R}^{d \times d}~~~~and~~~~
C=\left(\begin{array}{cccccc}
1 & 0 & 0 & \cdots & 0  & 0\\
0 & 0 & 0 & \cdots & 0  & 1\\
0 & 0 & 0 & \cdots & 1  & 0\\
\vdots & \vdots & \vdots &  \ddots & \vdots  & \vdots\\
0 & 1 & 0 & \cdots & 0  & 0 \end{array}\right)\in \mathbb{R}^{d \times d}. $$

Let $\mathrm{S}=(I_k~~0_{k\times (d-k)})$.  By multiplying matrix $S$ at
the left-hand side of an arbitrary matrix, one obtains its first $k$ rows as a new
rectangular matrix with dimension $k\times d$; thus the matrix
$Y$ can be written in the form
\begin{equation}
Y=
\sum_{i=0}^{d-1}\varkappa_ix_i\mathrm{S}P^iC\triangleq
\sum_{i=0}^{d-1}\varkappa_iB_i,
\end{equation}
where the random matrix $B_i=x_i\mathrm{S}P^iC$.
Now, we estimate the spectral norm of random matrix $Y$ by using Lemma \ref{lem1}.
\begin{lemma}\label{lem2}
Let $\mathrm{Y}$ be defined as before. Then it holds
\begin{equation}\label{equ3}
\mathbb{P} \{\|\mathrm{Y}\|\geq t \}\leq (d+k)e^{-t^2/2}.
\end{equation}
\end{lemma}

\begin{proof}
By Lemma \ref{lem1}, we only need to show that
$$\max\{\|\sum_{i=0}^{d-1}B_iB^T_i\|,\|\sum_{i=0}^{d-1}B^T_iB_i\|\}=1,$$
where $B_i=x_i\mathrm{S}P^iC$. In fact, on one
hand,
\begin{equation}
\sum_{i=0}^{d-1}B_iB^T_i=\sum_{i=0}^{d-1} x_i^2
\mathrm{S}\mathrm{P}^i\mathrm{C}\mathrm{C}^T(\mathrm{P}^i)^T\mathrm{S}^T=\sum_{i=0}^{d-1}
x_i^2 I_k,
\end{equation}
where we have employed the property $Q^T=Q^{-1}$ for
every permutation $Q$, $\mathrm{S}\mathrm{S}^T=I_k$, $\mathrm{C}\mathrm{C}^T=I_d$, and $\mathrm{P}^i(\mathrm{P}^i)^T=I_d$. Since vector
$x$ is a unit vector, we get $\sum_{i=0}^{d-1}B_iB^T_i=I_k$ which implies
$\|\sum_{i=0}^{d-1}B_iB^T_i\|=1$. On the other hand,

\begin{subequations}
\begin{align}
\sum_{i=0}^{d-1}B^T_iB_i &= \sum_{i=0}^{d-1}x_i^2\mathrm{C}^T(\mathrm{P}^i)^T\mathrm{S}^T\mathrm{S}\mathrm{P}^i\mathrm{C} \\
               &= \sum_{i=0}^{d-1}x_i^2\mathrm{C}^T(\mathrm{P}^i)^T
\left(\begin{array}{cc}
I_k & 0 \\
0  & 0 \end{array}\right)
               \mathrm{P}^i\mathrm{C} \\
               &\preceq  \sum_{i=0}^{d-1}x_i^2 I_d= I_d.
\end{align}
\end{subequations}
Thus, $\|\sum_{i=0}^{d-1}B^T_iB_i\|\leq 1$. This completes the proof.
\end{proof}

Taking $t=\sqrt{\tau}\log^{\delta/2} n$ with $\delta, \tau$ being positive
parameters in the probability inequality (\ref{equ3}), we have the
following estimation
\begin{equation} \label{equ4}
\|\mu\|_\infty=\|\lambda\|^2_\infty=\|\mathbf{Y}\|^2\leq
\tau\log^\delta n
\end{equation}
with probability at least $1-(d+k)e^{-\frac{\tau\log^\delta n}{2}}$.
From (\ref{equ1}) and (\ref{equ4}), the following holds
\begin{equation} \label{equ5}
\|\mu\|_2\leq \sqrt{\|\mu\|_1\|\mu\|_\infty}\leq \sqrt{\tau k\log^\delta n}
\end{equation}
with probability at least $1-(d+k)e^{-\frac{\tau\log^\delta n}{2}}$.

\textbf{Step 3: Concentration.}
To finish the proof, we need the following concentration result \cite{Laurent}, which is also the main tool employed in \cite{Hinrichs,Vyb}.
\begin{lemma}\label{lem}
Let $Z=\sum_{i=1}^s \alpha_i(a_i^2-1)$ where $a_i$ are independent identically distributed (i.i.d.) normal
variables and $\alpha_i$ are nonnegative numbers. Then for any
$t>0$,
\begin{equation}\label{equ6}
\mathbb{P}(Z\geq 2\|\alpha\|_2\sqrt{t}+2\|\alpha\|_\infty t)\leq \exp(-t),
\end{equation}
\begin{equation}\label{equ7}
\mathbb{P}(Z\leq -2\|\alpha\|_2 \sqrt{t})\leq \exp(-t).
\end{equation}
\end{lemma}

Now, let us complete the proof. First, we have
\begin{equation}
\mathbb{P}(\|Ya\|^2_2\geq (1+\epsilon)k) =\mathbb{P}(\sum_{j=0}^{k-1}\mu_j (b_j^2-1)\geq \epsilon k).
\end{equation}
Denote $Z=\sum_{j=0}^{k-1}\mu_j (b_j^2-1)$; then we need to estimate $\mathbb{P}(Z \geq k\epsilon)$.  By the estimation (\ref{equ6}) in Lemma \ref{lem}, we get
\begin{equation}\label{equ8}
\mathbb{P}(Z\geq 2\|\mu\|_2\sqrt{t}+2\|\mu\|_\infty t)\leq \exp(-t).
\end{equation}
Using (\ref{equ4}) and (\ref{equ5}), we derive
\begin{equation}
\mathbb{P}(Z\geq 2\sqrt{\tau kt\log^\delta n}+2\tau t\log^\delta n )\leq \exp(-t).
\end{equation}
Setting $t=\frac{c(\tau)k\epsilon^2}{\log^{\delta} n}$ with $c(\tau)=\frac{1}{8\tau}$,
we have
\begin{equation}
2\sqrt{\tau kt\log^\delta n}+2\tau t\log^\delta n = (\frac{\sqrt{2}}{2}+\frac{\epsilon}{4}) k\epsilon\leq k\epsilon.
\end{equation}
 Thus, we finally get
\begin{equation}
\mathbb{P}(Z\geq k\epsilon)\leq \exp(-\frac{c(\tau)k\epsilon^2}{\log^{\delta} n}),
\end{equation}
which shows (\ref{inequ1}). The inequality (\ref{inequ2}) can be proved in the same manner by invoking the estimation (\ref{equ7}) in Lemma \ref{lem}.

\section{Proof of Theorem \ref{lem01}}
For the subgaussian case, we provide a direct proof by using the Laplace transform technique. First, we need the following lemma:
\begin{lemma}\label{lem3}
If $X$ is subgaussian with constant $\eta >0$ and $X_i\sim X$ are i.i.d., then
\begin{equation}\label{equ9}
E[\exp(\lambda W^2)]\leq \frac{1}{\sqrt{1-4\eta\lambda}},
\end{equation} where
$W=\sum^{k-1}_{i=0}X_i\beta_i$ with $\beta_i\in \mathbb{R}$ satisfying $\sum^{k-1}_{i=0}\beta_i^2=1$. Moreover, define
$\varphi(\lambda)=\log E[\exp(\lambda (W^2-1))]$. Then it holds
\begin{equation}\label{equ10}
\varphi(\lambda)\leq \frac{8\eta^2\lambda^2}{1-4\eta \lambda}, ~~for~~\lambda<\frac{1}{4\eta}~~and~~\eta \leq \frac{1}{2}.
\end{equation}
\end{lemma}

\begin{proof}
For the proof of the first part, see \cite{Lee}. Here we only show the second part. Using the estimate of bound $E[\exp(\lambda W^2)]$ in (\ref{equ9}) and the conditions $\lambda<\frac{1}{4\eta}$ and $\eta \leq \frac{1}{2}$, we calculate
\begin{subequations}
\begin{align}
\varphi(\lambda) &\leq  -\frac{1}{2}\log(1-4\eta\lambda)-\lambda \\
                 &=     2\eta\lambda-\lambda+\sum_{m=2}^{\infty}2^{m-1}\frac{(2\eta\lambda)^m}{m} \\
                 &\leq \sum_{m=2}^{\infty}2^{m-1}\frac{(2\eta\lambda)^m}{m}=\sum_{m=2}^{\infty} \frac{8\eta^2\lambda^2(4\eta\lambda)^{m-2}}{m}\label{sum1}\\
                 &\leq  \sum_{m=2}^{\infty}8\eta^2\lambda^2(4\eta\lambda)^{m-2} = \sum_{m=0}^{\infty}8\eta^2\lambda^2(4\eta\lambda)^{m}\\
                 &=\frac{8\eta^2\lambda^2}{1-4\eta \lambda}, ~~for~~\lambda<\frac{1}{4\eta}~~and ~~\eta \leq \frac{1}{2}\label{sum2}
\end{align}
\end{subequations}
which completes the proof.
\end{proof}

We divide the proof of Theorem \ref{lem01} into two parts.

\textbf{Part A: Proof of probability inequality (\ref{inequ01}).}
Similar to the argument in the proof of Theorem \ref{lem0}, we need to estimate
\begin{equation}
\mathbb{P}(\sum_{j=0}^{k-1}\mu_j (b_j^2-1)\geq \epsilon k),
\end{equation}
where $b_j= \sum_{i=0}^{d-1}V_{ij}a_i$ is not a Gaussian variable but a linear combination of subgaussian variables, i.e., each $b_j$ has the form of $W$ in Lemma \ref{lem3}. So we can't directly invoke Lemma \ref{lem2}.
Here we use the Laplace transform technique to complete the proof. We derive that
\begin{subequations}
\begin{align}
\mathbb{P}(\sum_{j=0}^{k-1}\mu_j (b_j^2-1)\geq \epsilon k)
                 &=      \mathbb{P}(\exp(\sum_{j=0}^{k-1}\lambda\mu_j (b_j^2-1))\geq \exp(\lambda\epsilon k)), ~~for ~~ \lambda>0 \\
                 &\leq  \inf_{\lambda>0}\frac{E[\exp(\sum_{j=0}^{k-1}\lambda\mu_j (b_j^2-1))]}{\exp(\lambda\epsilon k)} \label{Mar} \\
                 &=  \inf_{\lambda>0}\frac{\prod_{j=0}^{k-1}E[\exp(\lambda\mu_j (b_j^2-1))]}{\exp(\lambda\epsilon k)} \label{ind}\\
                 &\leq  \inf_{0<\lambda\mu_j<1/4\eta}\frac{\prod_{j=0}^{k-1}\exp(\varphi(\lambda\mu_j))}{\exp(\lambda\epsilon k)},\label{est}
\end{align}
\end{subequations}
where (\ref{Mar}) follows from the Markov inequality, (\ref{ind}) follows from the independence of $b_j$, and (\ref{est}) is due to the additional restriction of $\lambda$ and the expression of $\varphi(\cdot)$. Denote $f(\lambda)=\frac{8\eta^2\lambda^2}{1-4\eta \lambda}$; then it is a monotonically increasing function since its derivative is positive. Moreover $\|\mu\|_\infty\leq \tau \log^\delta n$ with probability at least $1-(d+k)e^{-\frac{\tau\log^\delta n}{2}}$ from (\ref{equ5}). Thus, together with Lemma \ref{lem3} we have
\begin{equation}
\varphi(\lambda\mu_j)\leq f(\lambda\mu_j)\leq \frac{8\eta^2\lambda^2\tau^2\log^{2\delta}n}{1-4\eta\lambda\tau\log^\delta n},~~for~~ j=0,\cdots, k-1.
\end{equation}
With this uniform bound and a tighter restriction of $\lambda$, we continue to estimate the probability inequality (\ref{est}) and get:
\begin{eqnarray}
\mathbb{P}(\sum_{j=0}^{k-1}\mu_j (b_j^2-1)\geq \epsilon k)
                 &\leq & \inf_{0<\lambda<1/4\eta\tau\log^\delta n}\exp\left(\frac{8k\eta^2\lambda^2\tau^2\log^{2\delta}n}{1-4\eta\lambda\tau\log^\delta n}-\lambda\epsilon k\right).
\end{eqnarray}
Take $\lambda=\frac{\theta\epsilon}{2\eta\tau \log^{2\delta} n}$, where $\theta$ is a positive parameter and $\epsilon$ obeys $0<\epsilon<1/2$. In order to satisfy the constraint $0<\lambda<(4\eta\tau\log^\delta n)^{-1}$, one needs to require that $0<\theta<1$. Now, using this special choice of $\lambda$, we get an upper bound
\begin{equation}
\mathbb{P}(\sum_{j=0}^{k-1}\mu_j (b_j^2-1)\geq \epsilon k)\leq \exp\left(-\frac{k\epsilon^2}{\log^{2\delta}n}(\frac{\theta}{2\eta\tau}-\frac{2\theta^2}{1-\frac{2\theta\epsilon}{\log^\delta n}})\right).
\end{equation}
For any fixed parameter $\delta$, let $n$ be big enough such that $\frac{2\theta\epsilon}{\log^\delta n}<\frac{1}{2}$. Then the upper bound can be relaxed to
\begin{equation}\label{equ11}
\mathbb{P}(\sum_{j=0}^{k-1}\mu_j (b_j^2-1)\geq \epsilon k)\leq \exp\left(-\frac{k\epsilon^2}{\log^{2\delta}n}(\frac{\theta}{2\eta\tau}-4\theta^2)\right).
\end{equation}
Let $c(\theta,\eta, \tau)=\theta(\frac{1}{2\eta\tau}-4\theta)$; then it is an positive constant depending on parameters $\theta, \eta, \tau$ if $\theta<\frac{1}{8\eta\tau}$. Thus, the probability inequality (\ref{inequ01}) holds.

\textbf{Part B: Proof of probability inequality (\ref{inequ02}).} In the following, we will show that the inequality (\ref{inequ02}) can be obtained in the same manner under the additional parameters constraint $\frac{1}{2}\frac{1-\beta^2}{1+\beta^2}\leq \eta \leq \frac{1}{2}$ where $\beta=\frac{\theta\epsilon}{\tau\log^{2\delta}n}<\frac{1}{2}$. Our aim is to estimate
\begin{equation}
\mathbb{P}(\sum_{j=0}^{k-1}\mu_j (1-b_j^2)\geq \epsilon k).
\end{equation}
Define a new function
\begin{equation}
\phi(\lambda)=\log E[\exp(\lambda (1-W^2))],
\end{equation}
where the random variable $W$ is defined as in Lemma \ref{lem3}.
Applying the Laplace transform technique and using the new function above, we get
\begin{subequations}
\begin{align}
\mathbb{P}(\sum_{j=0}^{k-1}\mu_j (1-b_j^2)\geq \epsilon k)
                 &\leq       \inf_{\lambda>0}\frac{E[\exp(\sum_{j=0}^{k-1}\lambda\mu_j (1-b_j^2))]}{\exp(\lambda\epsilon k)}  \\
                 &\leq  \inf_{0<\lambda\mu_j<1/4\eta}\frac{\prod_{j=0}^{k-1}\exp(\phi(\lambda\mu_j))}{\exp(\lambda\epsilon k)}.
\end{align}
\end{subequations}
If we could prove the following inequality
\begin{equation}\label{equ12}
\phi(\lambda)\leq \frac{8\eta^2\lambda^2}{1-4\eta \lambda}, ~~when~~2\lambda\eta=\beta<\frac{1}{2}~~and~~ \frac{1}{2}\frac{1-\beta^2}{1+\beta^2}\leq \eta \leq \frac{1}{2},
\end{equation}
where $\beta=\frac{\theta\epsilon}{\tau\log^{2\delta}n}$, then we can prove (\ref{inequ02}) as \textbf{Part A} because setting $\lambda=\frac{\theta\epsilon}{2\tau\eta \log^{2\delta} n}$, $\varphi(\lambda)$ and $\phi(\lambda)$ take the same upper bound.
Now, let us show inequality (\ref{equ12}) as follows:
\begin{subequations}
\begin{align}
\phi(\lambda) &\leq  -\frac{1}{2}\log(1+4\eta\lambda)+\lambda \label{addine} \\
                 &= -\frac{1}{2}\sum_{m=1}^{\infty}(-1)^{m-1}\frac{(4\eta\lambda)^m}{m}+\lambda \\
                 &= \frac{1}{2}\sum_{m=1}^\infty\frac{(4\eta\lambda)^m}{m} - \sum_{m\equiv1(mod 2)}^\infty \frac{(4\eta\lambda)^m}{m} + \lambda \\
                 &= \sum_{m=2}^{\infty}2^{m-1}\frac{(2\eta\lambda)^m}{m} +2\eta\lambda- \sum_{l=1}^{\infty}\frac{(4\eta\lambda)^{2l-1}}{2l-1}+\lambda,
\end{align}
\end{subequations}
where (\ref{addine}) follows from the first part of Lemma \ref{lem3}.
From (\ref{sum1}) to (\ref{sum2}), it holds under the condition $\lambda\eta<\frac{1}{4}$ that
\begin{equation}
\sum_{m=2}^{\infty}2^{m-1}\frac{(2\eta\lambda)^m}{m}\leq\frac{8\eta^2\lambda^2}{1-4\eta \lambda}.
\end{equation}
Denote $g(\lambda)=\sum_{l=1}^{\infty}\frac{(4\eta\lambda)^{2l-1}}{2l-1}$. Then
\begin{equation}
g(\lambda)=\sum_{l=1}^{\infty} 2\frac{2^{2l-2}}{2l-1}(2\eta\lambda)^{2l-1}\geq 2\sum_{l=1}^{\infty} (2\eta\lambda)^{2l-1}=\frac{4\eta\lambda}{1-4\eta^2\lambda^2},
\end{equation}
where the inequality follows from that $\frac{2^{2l-2}}{2l-1}\geq 1$ for every positive number $l$.
Hence, it suffices to show $2\eta\lambda-\frac{4\eta\lambda}{1-4\eta^2\lambda^2}+\lambda\leq 0$, or equivalently to show $\beta-\frac{2\beta}{1-\beta^2}+\frac{\beta}{2\eta}\leq 0$ since $2\lambda\eta=\beta$. After a simple calculation, one needs $\frac{1}{2}\frac{1-\beta^2}{1+\beta^2}\leq \eta$ which is just the assumed condition. Thus, the inequality (\ref{equ12}) holds and hence the estimate (\ref{inequ02}) follows.

\begin{remark}
The condition on the subgaussian constant $\frac{1}{2}\frac{1-\beta^2}{1+\beta^2}\leq \eta$ is only required in estimating (\ref{inequ02}). Such a requirement can guarantee inequality (\ref{equ12}) hold and hence gives us a uniform probability estimates. If one gives up the uniform expressions in (\ref{inequ01}) and (\ref{inequ02}), then the condition may be relaxed. We leave the possible improvements of the lower bound on the subgaussian constant open for further investigations.
\end{remark}

\small{
{\bf Acknowledgement.}
\medskip
 
 The authors thank the anonymous reviewers for their valuable comments and referring us to \cite{KMR} that greatly improve the quality of the paper, and thank Profs. Jan Vyb\'{\i}ral and Wotao Yin for their helpful discussion, and also thank Dr. Ming Yan for a careful reading of the manuscript.
The work of H. Zhang is supported by China Scholarship Council and the Graduate School of NUDT under Funding of Innovation B110202, and Hunan Provincial Innovation Foundation For Postgraduate CX2011B008.  The work of L.Z. Cheng is supported by the National Science Foundation of China under Grants No. 61271014 and No.61072118.  H. Zhang thanks Rice University, CAAM Department, for hosting him.

}
\end{document}